\documentclass[a4paper,USenglish,cleveref, autoref, thm-restate]{lipics-v2021}

\title{Graph Reconstruction with a Connected Components Oracle} 

\titlerunning{Graph Reconstruction with a Connected Components Oracle} 

\author{Juha Harviainen}{Department of Computer Science, University of Helsinki, Finland}{juha.harviainen@helsinki.fi}{https://orcid.org/0000-0002-4581-840X}{}

\author{Pekka Parviainen}{Department of Informatics, University of Bergen, Norway}{pekka.parviainen@uib.no}{}{This work was supported by L. Meltzers Høyskolefond.}

\authorrunning{J. Harviainen and P. Parviainen} 

\Copyright{Juha Harviainen and Pekka Parviainen} 



\category{} 

\relatedversion{} 

\nolinenumbers 

\hideLIPIcs  

\EventEditors{John Q. Open and Joan R. Access}
\EventNoEds{2}
\EventLongTitle{42nd Conference on Very Important Topics (CVIT 2016)}
\EventShortTitle{CVIT 2016}
\EventAcronym{CVIT}
\EventYear{2016}
\EventDate{December 24--27, 2016}
\EventLocation{Little Whinging, United Kingdom}
\EventLogo{}
\SeriesVolume{42}
\ArticleNo{23}

\usepackage{graphicx} 
\usepackage{amsmath}
\usepackage{amsthm}
\usepackage{amssymb}
\usepackage{xcolor}
\usepackage{mathtools}
\usepackage{booktabs}
\usepackage{makecell}
\usepackage{tcolorbox}
\usepackage{xspace}

\theoremstyle{definition}
\newtheorem{fact}{Fact}

\usepackage[ruled,noend,linesnumbered]{algorithm2e}
\SetKwInput{KwData}{Input}
\SetKwInput{KwResult}{Output}
\SetKw{Continue}{continue}
\SetKw{Break}{break}

\usepackage[textsize=scriptsize,backgroundcolor=green!10]{todonotes}

\newcommand{\OO}{\mathcal{O}}
\newcommand{\gr}{\textsc{\textup{GR}}\xspace}

\newcommand{\supgraph}[0]{G^+}

\DeclareMathOperator{\CC}{CC}
\DeclareMathOperator{\Sep}{Sep}
\DeclarePairedDelimiter\ceil{\lceil}{\rceil}

\newcommand{\defprob}[4]{
\begin{tcolorbox}[colback=gray!5!white,colframe=gray!75!black]
  \vspace{-1mm}
  \begin{tabular*}{\textwidth}{@{\extracolsep{\fill}}lr} #1  \\ \end{tabular*}
  {\bf{Input:}} #2  \\
  {\bf{Question:}} #3
  \vspace{-1mm}
\end{tcolorbox}
}

\begin{document}

\maketitle

\begin{abstract}
In the \textsc{Graph Reconstruction} (\gr) problem, the goal is to recover a hidden graph by utilizing some oracle that provides limited access to the structure of the graph. The interest is in characterizing how strong different oracles are when the complexity of an algorithm is measured in the number of performed queries. We study a novel oracle that returns the set of connected components (CC) on the subgraph induced by the queried subset of vertices. Our main contributions are as follows:
\begin{enumerate}
    \item For a hidden graph with $n$ vertices, $m$ edges, maximum degree $\Delta$, and treewidth $k$, \gr can be solved in $\OO(\min\{m/\log m, \Delta^2, k^2\} \cdot \log n)$ CC queries by an adaptive randomized algorithm.
    \item For a hidden graph with $n$ vertices and degeneracy $d$, \gr can be solved in $\OO(d^2 \log^2 n)$ CC queries by an adaptive randomized algorithm.
    \item For a hidden graph with $n$ vertices, $m$ edges, maximum degree $\Delta$, and treewidth $k$, no algorithm can solve \gr in $o(\min\{m, \Delta^2, k^2\})$ CC queries.
\end{enumerate}
\end{abstract}

\maketitle

\section{Introduction}

The \textsc{Graph Reconstruction} (\gr) problem of discovering the structure of an unknown graph has been investigated extensively from both theoretical \cite{Angluin08,Konrad24} and practical \cite{Chow68,Karger01,Koller09,Korhonen24} point of view with applications to, for example, bioinformatics \cite{Reyzin07}. There, we are given the set of vertices of a \emph{hidden graph} and access to an \emph{oracle} that can be \emph{queried} to obtain information about the structure of the graph.

The \gr problem has also been studied within the machine learning community. Specifically, structure learning of probabilistic graphical models~\cite{Koller09} is a special case of \gr. Probabilistic graphical models are representations of multivariate probability distributions. They use a graph structure to express dependencies among variables. The vertices of the graph correspond to the random variables of the model, and the relationships between variables are encoded by the edge structure of a graph. In structure learning, one is given samples from the (unknown) distribution and the goal is to reconstruct the graph structure of the data-generating distribution. 
One approach for learning this structure is then to utilize statistical conditional independence tests that state whether two variables are conditionally independent of each other when conditioned by some subset of other variables. Essentially, these tests can be seen as calls to a specific oracle depending on which particular graphical model is in question. For Markov networks, where the graph is undirected, one gives the oracle a triple $(u, v, S)$ and the oracle outputs whether the corresponding two vertices $u$ and $v$ would be in different connected components if the vertices for the variables $S$ were to be removed. With Bayesian networks, whose structure is a directed acyclic graph, the independencies are captured by a more involved structural notion of $d$-separation. Ideally, one would like to discover the structure of a model using only a moderate number of queries, possibly under constraints on the size of the queries or the structure of the hidden graph~\cite{Chow68,Karger01,Koller09,Korhonen24}.

The query complexity of \gr has been studied for varying oracles and families of graphs. One main line of research has focused on the independent set (IS) oracle that only reports whether the queried subset is an independent set. This oracle has been applied to learning matchings \cite{Alon04,Beigel01}, cycles \cite{Grebinski98}, stars \cite{Alon05}, cliques \cite{Alon05}, and general graphs \cite{Abasi19,Angluin08}. Other studied oracles include betweenness~\cite{Abrahamsen16,Rong22}, distance~\cite{Beerliova06,Kannan18,Mathieu23,Rong21}, edge counting~\cite{Bouvel05,Grebinski00,Reyzin07}, shortest path oracles~\cite{Reyzin07,Sen10}, and connected component counting oracle~\cite{Black25}.

From an information-theoretical point of view, many of the previous oracles give only few bits of information per query, and thus struggle in reconstructing larger families of graphs in a small number of queries. Konrad et al.~\cite{Konrad24} therefore recently proposed a more powerful oracle that returns a maximal independent set. Notably, their algorithms are able to reconstruct hidden graphs with bounded maximum degree in a logarithmic number of queries in the number of vertices because of the increased amount of information the oracle provides. The lower bounds for the query complexity with the oracle were later refined by Michel and Scott~\cite{Michel25}.

In this paper, we study the connected components (CC) oracle, which returns the sets of vertices in each connected component of the subgraph induced by the queried set. This is thus a stronger version of the connected component counting ($\#$CC) oracle of Black et al.~\cite{Black25}, which only returns the number of the connected components.
One can also see the CC oracle as a stronger version of the separation oracle used in learning Markov networks\footnote{To answer a $\Sep(a, b, S)$ query, one just conducts a $\CC(V\setminus S)$ query and check whether $a$ and $b$ are in the same connected component. Essentially, this means that a $\CC (V\setminus S)$ query answers to $\Sep(a, b, S)$ queries for all pairs $a, b \in V\setminus S$.}. We then investigate how powerful this CC oracle is as a function of classical graph parameters such as the maximum degree, the \emph{treewidth}, and the \emph{degeneracy} of the hidden graph, when the complexity is measured in terms of the number of performed queries. 

{\bf Desired properties of an algorithm.} There are several desirable properties for the algorithms, of which we focus on \emph{determinism} and \emph{non-adaptivity} similarly to previous works. An algorithm is deterministic if it does not utilize randomness, and otherwise it is \emph{randomized}. A deterministic algorithm should always give the correct output whereas a randomized algorithm has to succeed with \emph{high probability}, that is, the probability should tend to $1$ as the input size increases.
Randomized algorithms can often be made deterministic by \emph{derandomization} at the expense of slightly increased complexity by, for example, deterministically constructing a family of query sets with some combinatorial properties. 

Adaptivity refers on whether the performed queries can depend on the outcomes of the previous queries; a non-adaptive algorithm thus first constructs a set of (possibly randomized) queries, then performs all of them, and finally tries to recover the hidden graph based on the outcomes. The notion of having $r$ \emph{rounds of adaptivity} gives a more granular measure of adaptivity, where we are allowed to construct $r$ sets of queries such that the queries in the $i$th set can depend on the outcomes of the queries in the previous sets. Without any prior knowledge about the properties of the graph, non-adaptive algorithms tend to need a polynomial number of queries in the number of vertices, since the constructed set of queries needs to work correctly (with high probability) for all possible graphs.  

{\bf Our contributions.} We start by devising several algorithms for the \gr problem with a CC oracle. A summary of the upper bounds is shown in Table~\ref{table:ub}. We present results parameterized by four graph parameters: number of edges $m$, maximum degree $\Delta$, treewidth $k$, and degeneracy $d$. With all these parameters, we provide adaptive algorithms (both deterministic and randomized) whose query complexity is polynomial with respect to the parameter and (poly)logarithmic with respect to the number of vertices. When the maximum degree of the hidden graph is known, we also provide a non-adaptive algorithm. The algorithms for $m$, $\Delta$, and $k$ can be combined, giving a randomized adaptive algorithm that uses at most $\OO(\min\{m/\log m, \Delta^2, k^2\} \cdot \log n)$ queries (Theorem~\ref{thm:combination}).
Our algorithm for graphs with bounded degeneracy uses $\OO(d^2 \log^2 n)$ queries, so it is mildly worse in terms of $n$. However, the degeneracy of a graph is always the smallest one of all the considered parameters, and so this algorithm is applicable to many graph families for which the other algorithms require polynomially many queries in $n$. For example, planar graphs have degeneracy at most five~\cite{Lick70} but can have unbounded maximum degree and treewidth. 

We also study lower bounds; a summary is shown in Table~\ref{table:lb}. Black et al.\ \cite{Black25} showed that it is not possible to design efficient non-adaptive algorithms for the $\#$CC queries, and we note that their result applies also for CC queries. In other words, non-adaptive algorithms require $\Omega (n^2)$ CC queries in the worst case if no additional information is given (Corollary~\ref{cor:lower_bound_n}), even if we know that the graph has treewidth $2$ and thereby has at most $\OO (n)$ edges. There is a peculiar observation that very little adaptivity can help. Algorithm~\ref{alg:gr_tw} has two rounds of adaptivity and query complexity $\OO (k^2 \log n)$ implying that adding one round of adaptivity reduces the dependence on $n$ from $n^2$ to $\log n$.
Furthermore, we show that no adaptive algorithm can solve the problem for all instance by performing $o(m)$, $o(\Delta^2)$, $o(k^2)$, or~$o(d^2)$ CC queries (Theorem~\ref{thm:lower_bound_mkd}). Thus, our upper bounds are within a (poly)logarithmic factor in $n$ from the lower bound.

\begin{table}
    \centering
    \caption{Summary of our upper bounds. The first corollary follows directly from Black et al.~\cite{Black25}.} \label{table:ub}
    \begin{tabular}{ccccc}
        \toprule
        Theorem & Complexity & Randomized & Adaptive & Comments\\
        \midrule
        Corollary~\ref{thm:edge} & $\OO\big(m \log(n) / \log(m)\big)$ & no & yes & -\\
        Theorem~\ref{thm:max_deg_randomized} & $\OO(\Delta^2 \log n)$ & yes & no & known $\Delta$\\
        Theorem~\ref{thm:max_deg_deterministic} & $\OO\big(\Delta^3 \log(n / \Delta)\big)$ & no & no & known $\Delta$, existence result\\
        Theorem~\ref{thm:max_deg_unknown} & $\OO\big(\Delta^2 \log n \big)$ & yes & yes & -\\
        Theorem~\ref{thm:max_deg_unknown} & $\OO\big(\Delta^3 \log(n / \Delta)\big)$ & no & yes & existence result\\
        Theorem~\ref{thm:tw_unknown} & $\OO(k^2 \log n)$ & yes & yes & -\\
        Theorem~\ref{thm:tw_unknown} & $\OO(k^3 \log n)$ & no & yes &  existence result\\
        Theorem~\ref{thm:degeneracy_unknown} & $\OO(d^2 \log^2 n)$ & yes & yes & -\\
        Theorem~\ref{thm:degeneracy_unknown} & $\OO\big(d^3 \log n \log(n / d)\big)$ & no & yes &  existence result\\
        \bottomrule
    \end{tabular}
\end{table}

\begin{table}
\centering
    \caption{Summary of our lower bounds, where the Complexity column states that there is an instance on which $\Omega(\cdot)$ CC queries are needed. A lower bound for randomized algorithms holds for deterministic algorithms, and a lower bound for adaptive algorithms holds for non-adaptive algorithms. The proof of Corollary~\ref{cor:lower_bound_n} is identical to the lower bound proof for the $\#$CC oracle \cite{Black25}.} \label{table:lb}
     \begin{tabular}{ccccc}
        \toprule
        Theorem & Complexity & Randomized & Adaptive & Comments\\
        \midrule
        Theorem~\ref{thm:lower_bound_mkd} & $\Omega (m)$ & yes & yes & -\\
        Theorem~\ref{thm:lower_bound_mkd} & $\Omega (k^2)$ & yes & yes & -\\
        Theorem~\ref{thm:lower_bound_mkd} & $\Omega (\Delta^2)$ & yes & yes & -\\
        Theorem~\ref{thm:lower_bound_mkd} & $\Omega (d^2)$ & yes & yes & -\\
        Corollary~\ref{cor:lower_bound_n} & $\Omega (n^2)$ & yes & no & even if $k=2$ and $m = \OO(n)$\\
        \bottomrule
    \end{tabular}
\end{table}

Finally, we compare the CC oracle with other oracles. First, we show that getting explicit information about the connected components rather than just their number \cite{Black25} leads to more efficient algorithms (Theorem~\ref{thm:counting-is-weak}). Next, we consider the maximal independent set (MIS) oracle introduced by Konrad et al.\ \cite{Konrad24}. While there are instances where the MIS oracle needs $\OO (n)$ queries but $\OO (\log n)$ queries are sufficient for the CC oracle (Theorem~\ref{thm:cc_powerful}), these two oracles seem to be roughly equally powerful in the sense that both upper and lower bounds for $m$ and $\Delta$ are nearly the same for both types of oracles. However, CC oracles and MIS oracles are different and neither of them can be efficiently simulated with the other (Theorems~\ref{thm:sim_cc} and~\ref{thm:sim_mis}). We also compare the CC oracle and the separation oracle, showing that simulating the former oracle with the latter one requires a quadratic number of queries in the size of the queried set (Theorem~\ref{thm:sep_ub}).

\section{Notation and Preliminaries}

An {\em undirected graph} is a pair $G = (V(G), E(G))$ where $V(G)$ is the vertex set of $G$ and $E(G)$ is the edge set of $G$. If there is no ambiguity about which graph we refer to, we simply write $G = (V, E)$. An edge between vertices $u$ and $v$ is denoted by $uv$; note that edges are undirected and ordering does not matter, that is, $uv$ and $vu$ are the same edge. We use $n = |V(G)|$ to denote the number of vertices and $m = |E(G)|$ to denote the number of edges in a graph. If there is no edge between vertices $u$ and $v$, we call $uv$ a \emph{non-edge}. 

Vertex $u$ is a neighbor of $v$ in $G$ if there is an edge $uv \in E(G)$. We denote the set of neighbors in $G$ by $N_G(v)$ or simply by $N(v)$ if the graph $G$ is clear from the context. The degree of a vertex $v$ is the number of its neighbors. The maximum degree $\Delta$ of a graph is the maximum out of the degrees of its vertices.

A graph $H$ is a {\em subgraph} of $G$ if $V(H) \subseteq V(G)$ and $E(H) \subseteq E(G)$.
A graph $J$ is a {\em supergraph} of $G$ if $V(G) \subseteq V(J)$ and $E(G) \subseteq E(J)$.
In this paper, all used supergraphs $J$ of $G$ have $V(J) = V(G)$.

The subgraph of $G$ \emph{induced by} $S \subseteq V$, denoted by $G[S]$, has the vertex set $S$ and the edge set
$\big\{uv \in E(G) \colon u, v \in S\big\}$.
A subset of vertices is an \emph{independent set} if the subgraph induced by it has no edges. An independent set is \emph{maximal} if none of its proper supersets are independent.
The \emph{degeneracy} of a graph $d = d(G)$ is the smallest integer such that every induced subgraph has at least one vertex with degree at most $d$.

A {\em connected components} (CC) \emph{oracle} is a black box subroutine that is given a subset $S$ of vertices and returns the set of connected components of $G[S]$, that is, a partition of $S$ such that each subset is a connected component. 

In the \textsc{Graph reconstruction} problem, we are given a vertex set and an oracle, and asked to perform queries on the oracle to reconstruct the edge structure of the graph:

\defprob{\textsc{Graph reconstruction (GR)}}{A vertex set $V$, an oracle $O$}{What is the edge set $E$ of the hidden graph $G = (V, E)$?}{}

In this paper, we focus on using a CC oracle as the oracle and use the number of performed queries to measure the complexity of the algorithms for the \gr problem.

Throughout the paper, we will denote by $\supgraph$ the supergraph of $G$, which is initialized to a complete graph and from which edges are removed as the queries reveal more information.

{\bf Treewidth.} A {\em tree} is a connected graph with no cycles. A \emph{tree decomposition} of a graph $G$ is a pair $(T, \mathcal{X})$, where $T$ is a tree and $\mathcal{X} = \{X_t \colon t \in V(T)\}$ a family of subsets of $V$ called \emph{bags} with the following properties:
\begin{enumerate}
    \item For every vertex $v \in V(G)$, there is a bag $X_t$ with $v \in X_t$;
    \item For every edge $uv \in E(G)$, there is a bag $X_t$ with $u, v \in X_t$; and
    \item For every vertex $v \in V(G)$, the subgraph of $T$ induced by its vertices $t$ with $v \in X_t$ is connected.
\end{enumerate}
The \emph{width of a tree decomposition} is the size of its largest bag minus one. The \emph{treewidth of a graph} is the smallest width among its tree decompositions. We will denote the treewidth of the hidden graph $G$ by $k$.
A maximal graph with treewidth $k$ is called a {\em $k$-tree}.

A pair of sets of vertices $(A, B)$ with $A \cup B = V(G)$ is called a {\em separation} if there are no edges between $A \setminus B$ and $B\setminus A$. The set $A\cap B$ is then called a {\em separator} and its {\em order} is $|A\cap B|$. In other words, removal of the separator $A \cap B$ would partition the graph into two separate connected components $G[A\setminus B]$ and $G[B\setminus A]$. A separator $S$ is {\em $\alpha$-balanced} if the size of every connected component in $G[V \setminus S]$ is at most $\alpha\cdot |V(G)|$.

Let $f(n)$ be a function from positive integers to nonnegative real numbers. Recall that $f(n) = o(g(n))$ for some function $g$ if $f(n)/g(n)$ tends to $0$ as $n$ tends to infinity, and $f(n) = \omega(g(n))$ if $f(n)/g(n)$ tends to infinity.

The following lemma will be useful later. Various versions of it have appeared in the literature for constructing sets that include and exclude subsets of elements.
\begin{lemma} \label{lemma:selection}
    Let $V$ be a set and let $U, W \subseteq V$ its disjoint subsets of size $q$ and $p$, respectively. Suppose that we sample a set $Q \subseteq V$ such that $v \in Q$ with probability $1/(ap + b)$ independently of each other, where $a > 0$ and $b \ge 0$ are constants with $ap + b > 1$ for all $p \ge 1$.
    Then, the probability that all elements of $U$ are in $Q$ and none of elements of $W$ are in $Q$ is at least $\left(\frac{1}{ap+b}\right)^q \cdot \min\{\mathrm{e}^{-1/a}, \mathrm{e}^{-1/(a+b-1)}\}$.
    In particular, the probability is $\Omega(1/p^q)$ for constant $q$.
\end{lemma}
\begin{proof}
    We have that
    \begin{eqnarray*}
        \Pr(U \subseteq Q \text{ and } W \cap Q = \emptyset) & = & \left(\frac{1}{ap+b}\right)^q \left(1 - \frac{1}{ap + b}\right)^p\\
        & \geq & \left(\frac{1}{ap+b}\right)^q \cdot \left(\mathrm{e}^{-1 / (ap + b - 1)}\right)^p\\
        & = & \left(\frac{1}{ap+b}\right)^q \cdot \mathrm{e}^{-p / (ap + b - 1)},
    \end{eqnarray*}
    where the inequality follows from the fact that $1 - 1/x \geq \mathrm{e}^{-1/(x - 1)}$ for every $x\geq 1$. 
    The derivative of $-p/(ap+b-1)$ with respect to $p$ is $(1 - b) / (ap + b - 1)^2$, which is zero if and only if $b = 1$, in which case $\mathrm{e}^{-p/(ap+b-1)}$ is $\mathrm{e}^{-a}$ everywhere. For other values of $b$, $\mathrm{e}^{-p/(ap+b-1)}$ is minimized either when $p = 1$ or $p$ tends to infinity, corresponding to values $\mathrm{e}^{-1/(a+b-1)}$ and $\mathrm{e}^{-1/a}$, respectively. 
\end{proof}

\section{Algorithms}

In this section, we will provide upper bounds for \gr with a CC oracle.

\subsection{Bounded Number of Edges}

Recently, Black et al.~\cite{Black25} showed an algorithm utilizing an oracle that returns the number of connected components for the subgraph induced by the query set $S$ can reconstruct any hidden graph on $m$ edges with $\OO\big(m \log(n) / \log(m)\big)$ queries. Since our oracle can trivially simulate theirs, the same upper bound follows.
\begin{corollary}\label{thm:edge}
    There is an adaptive deterministic algorithm that solves \gr with $\OO\left(\frac{m \log(n)}{\log(m)}\right)$ CC queries.
\end{corollary}

\subsection{Bounded Maximum Degree}

For graphs with a bounded maximum degree $\Delta$, we follow the strategy of Konrad et al.~\cite{Konrad24} of performing queries on randomized subsets of vertices. The main idea is that if there is no edge between $u$ and $v$, both vertices are included in a query set, and none of the neighbors of $u$ are included, then the CC query outputs them in different components.

\begin{algorithm}
\caption{An algorithm for \gr for graphs with bounded maximum degree.}\label{alg:max_deg_randomized}
\KwData{A vertex set $V$, a CC oracle $\CC(\cdot)$, maximum degree $\Delta$}
\KwResult{Hidden graph $G$ with high probability}
Initialize $\supgraph$ to a complete graph on $V$\;
\For{$t = \OO(\Delta^2 \log n)$ \textup{times}}{
    Sample $Q \subseteq V$ s.t. $v \in Q$ with probability $1/(\Delta + 1)$ for all $v \in V$ independently\;
    $\mathcal{C} \gets \CC(Q)$\;
    For all $u, v\in Q$, remove edge $uv$ from $\supgraph$ if $v$ and $u$ are in different components in $\mathcal{C}$\; 
}
\Return $\supgraph$\;
\end{algorithm}

\begin{theorem}\label{thm:max_deg_randomized}
     Algorithm~\ref{alg:max_deg_randomized} is a non-adaptive randomized algorithm that solves \gr with $\OO(\Delta^2 \log n)$ CC queries when $\Delta$ is known. The output of the algorithm is correct with high probability.
\end{theorem}
\begin{proof}
    Consider an arbitrary non-edge $uv$ of $G$. By Lemma~\ref{lemma:selection}, the probability that $Q$ contains $u$ and $v$ but none of the neighbors of $u$ is at least $1/((\Delta + 1)^2 \cdot \mathrm{e})$. If this occurs, then the output of the CC query must have $v$ and $u$ in distinct connected components. Let $X_i$ be a random Boolean variable that is true if this occurs at least once for the $i$th edge over $t$ independently chosen subsets $Q$, and otherwise false. By the union bound, we have that
    \begin{align*}
        \Pr\left(\bigwedge_{i=1}^{m} X_i\right)
        &\ge 1 - \sum_{i=1}^{m} \Pr(\overline{X_i})\\
        &\ge 1 - m \cdot \left(1 - \frac{1}{(\Delta + 1)^2 \cdot \mathrm{e}}\right)^t.
    \end{align*}

    We want to choose a value for $t$ such that $m \cdot \left(1 - 1/((\Delta + 1)^2 \cdot \mathrm{e})\right)^t \le 1 / n$, which holds if $\left(1 - 1/((\Delta + 1)^2 \cdot \mathrm{e})\right)^t \le 1 / n^3$. Since $(1 - 1/x)^x \le 1/\mathrm{e}$ for $x \ge 1$, choosing $t = \OO(\Delta^2 \log n)$ suffices for finding all edges with high probability.
\end{proof}

\subsubsection{Derandomization}

Next, we will derandomize Algorithm~\ref{alg:max_deg_randomized}. Our schemes follow straightforwardly from the one presented by Konrad et al.\ \cite{Konrad24}.

In Algorithm~\ref{alg:max_deg_randomized}, we remove an edge $uv$ if $u$ and $v$ are in different connected components returned by the query $\CC(Q)$ for some $Q$. If $u$ and $v$ are not connected by an edge in $G$, then there exists a set $\{ w_1, \ldots, w_p\}$ such that $u$ and $v$ are in different connected components in $\CC(Q)$ for any $Q$ where $u, v \in Q$ and $w_i \not\in Q$ for all $i$. Following Konrad et al.\ \cite{Konrad24}, we define a {\em witness} as follows.
\begin{definition}[Witness] \label{def:witness}
    Let $V$ be a set of $n$ vertices and $0 \leq p \leq n-2$ be an integer. Then, the tuple $(\{u, v\}, \{ w_1, \ldots, w_p\})$ with $u, v, w_1, \ldots, w_p \in V$ being distinct vertices is called a witness, and we denote the set of all witnesses by $\mathcal{W}$.
\end{definition}

The following lemma gives an upper bound for the number of witnesses.
\begin{lemma} \label{lemma:witness_count}
    The number of witnesses $|\mathcal{W}|$ is bounded by
    \begin{eqnarray*}
        |\mathcal{W}| & \leq & n^2 \cdot \Big(\mathrm{e} \cdot \frac{n - 2}{p}\Big)^p
    \end{eqnarray*}
\end{lemma}
\begin{proof}
    We use the bound ${a \choose b} \leq (\mathrm{e} \cdot \frac{a}{b})^b$ and get
    \begin{eqnarray*}
        {n \choose 2} {n - 2 \choose p} & \leq & n^2 \cdot \Big(\mathrm{e} \cdot \frac{n - 2}{p}\Big)^p.
    \end{eqnarray*}
\end{proof}

\begin{definition}[$p$-Query-Scheme] \label{def:query_scheme}
    Let $V$ be a set of $n$ vertices and $0 \leq p \leq n-2$ be an integer. The set $\mathcal{Q} = \{ Q_1, \ldots, Q_\ell\}$ is a $p$-Query-Scheme of size $\ell$ if, for every witness $(\{u, v\}, \{ w_1, \ldots, w_p\}) \in \mathcal{W}$, there exists a query $Q_i \in \mathcal{Q}$ such that
    \begin{enumerate}
        \item $u, v\in Q_i$ and
        \item $\{w_1, \ldots, w_p\} \cap Q_i = \emptyset$.
    \end{enumerate}
\end{definition}

The following lemma gives an upper bound for the size of a $p$-Query-Scheme.
\begin{lemma}\label{lemma:k-query}
    There exists a $p$-Query-Scheme of size $\OO(p^3 \log n/p)$.
\end{lemma}
\begin{proof}
    Suppose we have a collection of $\ell$ sets $\mathcal{Q} = \{ Q_1, \ldots, Q_\ell \}$ with $Q_i\subseteq V$ for all $i$. Suppose further that each vertex has probability $1/(p+1)$ to be in the set $Q_i$ independently for each set. That is, $\Pr(v\in Q_i) = 1 / (p+1)$ for all $v$ and $i$.

    Next, we will show that if $\ell = \OO (p^3 \log (n/p) )$, then the probability that $\mathcal{Q}$ is a $p$-Query-Scheme is larger than $0$.

    By Lemma~\ref{lemma:selection}, the probability that a set $Q_i$ contains a witness $(\{u, v\}, \{ w_1, \ldots, w_p\})$ is at least $1/((p+1)^2 \cdot \mathrm{e})$.
    Then, the probability that none of the sets in $\mathcal{Q}$ contains the witness $(\{u, v\}, \{ w_1, \ldots, w_p\})$ is
    \begin{eqnarray*}
        \Pr(\{u, v\} \not\subseteq Q_i \text{ or } \{w_1,\ldots w_p\} \cap Q_i \neq \emptyset \text{ for all } i) & \leq & \left(1 - \frac{1}{(p+1)^2} \cdot \frac{1}{\mathrm{e}}\right)^{\ell}\\
        & \leq & \exp\left(- \frac{\ell}{(p+1)^2 \cdot \mathrm{e}}\right),
    \end{eqnarray*}
    where the latter inequality follows from the fact $1 + x \leq \mathrm{e}^x$.

    By union bound, the probability that there exists a witness that is not considered for some pair $u$ and $v$ is at most
    \begin{equation}
        n^2 \left(\mathrm{e} \cdot \frac{n - 2}{p}\right)^p \exp\left(- \frac{\ell}{(p+1)^2 \cdot \mathrm{e}}\right)\,.
    \end{equation}

    To guarantee that the above expression is smaller than $1$, it suffices to choose $\ell$ such that
    \begin{align*}
        n^2 \left(\mathrm{e} \cdot \frac{n -2}{p}\right)^p \exp\left(- \frac{\ell}{(p+1)^2 \cdot \mathrm{e}}\right) & \le 1 & \iff\\
        \exp\left(- \frac{\ell}{(p+1)^2 \cdot \mathrm{e}}\right) & \le \frac{p^p}{n^2 (\mathrm{e} \cdot (n-2))^p} & \iff\\
        - \frac{\ell}{(p+1)^2 \cdot \mathrm{e}} & \le \ln \left( \frac{p^p}{n^2 (\mathrm{e} \cdot (n-2))^p} \right) & \iff\\
        \ell & \ge (p + 1)^2 \cdot \mathrm{e} \cdot p \cdot \ln \left( \frac{n^{2/p} \cdot \mathrm{e} \cdot (n-2)}{p} \right).
    \end{align*}

    For the last inequality to hold, it is enough to set $\ell = \Theta (p^3\log n/p)$.
\end{proof}

We are not aware of any efficient explicit constructions for $p$-Query-Schemes of size $\OO(p^3 \log n/p)$, but there is rich literature on closely related pseudorandom objects like \emph{group testing schemes} and \emph{strongly selective families}---see, for example, the work of Porat and Rothschild \cite{Porat11}.
At the cost of increased size, $(n, \kappa, \kappa^2)$-\emph{splitters} can be utilized to construct a $p$-Query-Scheme in polynomial time. Recall that an $(n, \kappa, \kappa^2)$-splitter for the vertices is a family~$\mathcal{F}$ of functions from the $n$ vertices to $\kappa^2$ colors $\{ 1, \dots, \kappa^2 \}$ such that for any subset of vertices $S$ of size $\kappa$, there is a function that colors the vertices of $S$ with distinct colors.
\begin{lemma}[\cite{Alon95}]\label{lemma:splitter}
    An $(n, \kappa, \kappa^2)$-splitter of size $\kappa^{\OO(1)} \log n$ can be constructed in time $\kappa^{\OO(1)} n \log n$.
\end{lemma}
\begin{theorem}
    A $p$-Query-Scheme of size $p^{\OO(1)} \log n$ can be constructed in time $p^{\OO(1)} n \log n$.
\end{theorem}
\begin{proof}
    Let $\kappa = p + 2$ and construct an $(n, \kappa, \kappa^2)$-splitter $\mathcal{F}$ of size $p^{\OO(1)} \log n$ in time $p^{\OO(1)} n \log n$ by applying Lemma~\ref{lemma:splitter}. Construct now in time $p^{\OO(1)} n \log n$ a family of sets
    \[ \mathcal{Q} = \left\lbrace f^{-1}(i) \cup f^{-1}(j) \colon f \in \mathcal{F}, i, j \in \{1, \dots, \kappa^2 \}\right\rbrace \]
    by considering each function and each pair of colors, and then including a set of vertices colored in either of the colors to $\mathcal{Q}$. We claim that $\mathcal{Q}$ is a $p$-Query-Scheme.

    Consider any witness $(\{u, v\}, \{w_1, \dots, w_p\})$. The splitter has a function $f$ that assigns a distinct color to each of the vertices of the witness. Letting, $i = f(u)$ and $j = f(v)$, the set $f^{-1}(i) \cup f^{-1}(j)$ contains the vertices $u$ and $v$ but none of the vertices $w_1, \dots, w_p$. Since this holds for all witnesses, $\mathcal{Q}$ is a $p$-Query-Scheme.
\end{proof}

By utilizing $\Delta$-Query-Schemes similarly to Konrad et al.\ \cite{Konrad24}, we get the following derandomized algorithm.
\begin{theorem}\label{thm:max_deg_deterministic}
    There exists a non-adaptive deterministic algorithm that solves \gr with $\OO\big(\Delta^3 \log(n / \Delta)\big)$ CC queries when $\Delta$ is known.
\end{theorem}
\begin{proof}
    We claim that it suffices to construct a $\Delta$-Query-Scheme $\mathcal{Q}$ and to perform those $t = \OO\big(\Delta^3 \log(n / \Delta)\big)$ queries in Algorithm~\ref{alg:max_deg_randomized} instead of the randomized queries.
    
    Consider any non-edge $uv$. By the definition of a $\Delta$-Query-Scheme, there exists a query $Q_i$ that contains both $u$ and $v$, but none of the neighbors of $u$. Thus, the oracle returns $u$ and $v$ in a separate component and the edge $uv$ gets removed from the supergraph~$\supgraph$.
\end{proof}

\subsubsection{Unknown $\Delta$}

By making the algorithms adaptive, we can solve \gr with the asymptotically same number of queries even when $\Delta$ is not given to us beforehand by following the strategy of Konrad et al.~\cite{Konrad24}. Essentially, we guess an upper bound for $\Delta$, and when we get evidence that it is incorrect, we double the value of the upper bound.

\begin{theorem}\label{thm:max_deg_unknown}
    There exist an adaptive randomized algorithm and an adaptive deterministic algorithm that solve \gr with $\OO\big(\Delta^2 \log n\big)$ and $\OO\big(\Delta^3 \log(n / \Delta)\big)$ CC queries, respectively. The output of the randomized algorithm is correct with high probability.
\end{theorem}
\begin{proof}
    We start by guessing that $D = 1$ is an upper bound for $\Delta$, and then run Algorithm~\ref{alg:max_deg_randomized} or the algorithm from Theorem~\ref{thm:max_deg_deterministic} as if the known value of $\Delta$ were $D + 1$. If $D$ is smaller than $\Delta$, then the graph outputted by that algorithm has a vertex whose degree exceeds $D$ by the construction of the query sets. In this case, we double the value of $D$ and run the algorithm again as if $\Delta$ were $D + 1$, repeating the process $\OO(\log \Delta)$ times.

    Each round of adaptivity performs $\OO\big(D^3 \log(n / D)\big)$ queries. Since $D$ is doubled each round, the total number of queries is dominated by the last round, resulting in the total number of queries being $\OO\big(\Delta^3 \log(n / \Delta)\big)$.

    For the randomized algorithm, we still need to show that the output is correct with high probability. Recall that  the probability of Algorithm~\ref{alg:max_deg_randomized} succeeds is at least $1 - 1/n$, and note that if the algorithm succeeds on all $\OO(\log \Delta)$ rounds of adaptivity, then the output is correct. The probability of this is at least $(1 - 1/n)^{\OO(\log n)} \approx (1/\mathrm{e})^{\OO(\log(n) / n)}$, which tends to $1$ as $n$ increases.
\end{proof}

Note that the type of adaptivity used by the above algorithms is rather weak: in a sense, we have a fixed (possibly randomized) sequence of queries where our guess for $\Delta$ varies, and the adaptivity is only used to decide when to stop performing more queries from the sequence.

Our algorithms parameterized by the maximum degree can be straightforwardly generalized for a recently introduced parameter \emph{maximum pairwise connectivity} $\lambda$, which is the maximum number of vertex-disjoint paths with at least one internal vertex between any pair of vertices \cite{Korhonen24}. Clearly, $\lambda$ is bounded by the maximum degree, since each of the disjoint paths has to use one of the edges adjacent to the starting vertex of the path. 
By Menger's theorem, any two non-adjacent vertices $u$ and $v$ can be separated by removing at most $\lambda$ vertices.
Thus, including the vertices to the queried sets with probability $1/(\lambda + 1)$ yields a tighter upper bound $\OO(\lambda^2 \log n)$ for the needed number of queries such that the algorithm succeeds with high probability. Similarly, we can consider $\lambda$-Query-Schemes for the derandomized version.

\subsection{Bounded Treewidth}

Next, we will provide algorithms for \gr in bounded treewidth graphs. In what follows, we assume that the hidden graph $G$ has treewidth $k$. We will use the following well-known fact. 

\begin{fact}[\cite{Rose74}] \label{fact:separators} 
    Let $G = (V, E)$ be a $k$-tree. Suppose $u, v\in V$ such that $uv\not\in E$. Then there exists a set $S\subseteq V \setminus \{u,v\}$ with $|S|= k$ such that $u$ and $v$ are in different connected components in $G[V\setminus S]$.
\end{fact}

We give an adaptive algorithm that uses 2 rounds of adaptivity when the treewidth $k$ of the hidden graph is known. First, the algorithm uses Fact~\ref{fact:separators} to construct a supergraph $\supgraph$ of the hidden graph $G$ such that the treewidth of $\supgraph$ is at most $k$. Then, the algorithm removes edges from $\supgraph$ to find $G$.

Algorithm~\ref{alg:tree_decomp} performs the first step: It outputs a supergraph of the hidden graph. Note that, to get a correct result, one does not actually need to know the treewidth of the hidden graph $G$ but an upper bound of it. However, the number of queries conducted increases if the bound is not tight. Next, we will prove the correctness and complexity of Algorithm~\ref{alg:tree_decomp}.

\begin{algorithm}
\caption{An algorithm for \gr for learning a supergraph of treewidth $k$.}\label{alg:tree_decomp}
\KwData{A vertex set $V$, a CC oracle $\CC(\cdot)$, upper bound for treewidth $k$}
\KwResult{A supergraph $\supgraph$ of the hidden graph $G$ with treewidth at most $k$ with high probability}
Initialize $\supgraph$ to complete graph on $V$\;
\For{$t = \OO(k^2 \log n)$ \textup{times}}{
    Sample $Q \subseteq V$ s.t. $v \in Q$ with probability $1/(k+1)$ for all $v \in V$ independently\;
    $\mathcal{C} \gets \CC(Q)$\;
    Remove edges $uv$ for which $v$ and $u$ are in different components from $\supgraph$\;
}
\Return $\supgraph$\;
\end{algorithm}

\begin{lemma}\label{lemma:tree_decomp}
    Algorithm~\ref{alg:tree_decomp} is a non-adaptive randomized algorithm that executes $\OO (k^2\log n)$ CC queries and returns a graph $\supgraph$ which is a supergraph of the hidden graph $G$ and, with high probability, has treewidth at most $k$.
\end{lemma}
\begin{proof}
    The proof is a straightforward adaptation from Konrad et al.\ \cite{Konrad24}.
    Let $\bar{G} = (V, \bar{E})$ be any $k$-tree which is a supergraph of $G$. Let $\bar{E}^- = V\times V \setminus \bar{E}$ be the set of non-edges in $\bar{G}$.
    Next, we will prove that Algorithm~\ref{alg:tree_decomp} detects all non-edges $\bar{E}^-$ of $\bar{G}$ with high probability.

    By Fact~\ref{fact:separators}, if $u$ and $v$ are not neighbors in $\bar{G}$ then there exists a separator $U$ of size $k$. Furthermore, as $\bar{G}$ is a supergraph of $G$, $U$ separates $u$ and $v$ in $G$. To detect the non-edge $uv$ it is sufficient to conduct query $\CC(Q)$ for some $Q$ subject to $u, v\in Q$ and $w\not\in Q$ for all $w\in U$.
    Again, by Lemma~\ref{lemma:selection}, the probability of this occurring is at least $1 / ((k + 1)^2 \cdot \mathrm{e})$.
    We conduct $t \coloneqq c \cdot (k + 1)^2 \cdot \ceil*{\ln n}$ queries. The probability that we fail to identify a non-edge is 
    \begin{eqnarray*}
        \Pr(\text{non-edge } uv \text{ not detected}) & = & \prod_{i = 1}^{t} \big(1 - \Pr(u, v\in Q \text{ and } w \not\in Q \text{ for all } w\in U)\big)\\
        & \leq & \left(1 - \frac{1}{(k + 1)^2 \cdot \mathrm{e}}\right)^{c \cdot (k + 1)^2 \cdot \ceil*{\ln n}} \\
        & \leq & \left( \mathrm{e}^{- \frac{1}{(k + 1)^2 \cdot \mathrm{e}}} \right)^{c \cdot (k + 1)^2 \cdot \ln n} \\
        & = & \mathrm{e}^{ - c \ln(n) / \mathrm{e}} \\
        & = & \left(\frac{1}{n}\right)^{c / \mathrm{e}}
    \end{eqnarray*}
    where the second inequality follows from the fact that $1 + x \leq \mathrm{e}^x$.
    If $c \geq 3\mathrm{e}$ then $\Pr(\text{non-edge } uv \text{ not detected}) \leq 1/n^3$.

    There are at most $n(n-1)/2$ potential non-edges. By union bound, the probability that at least one of them is not detected is at most $\OO(1/n)$.
\end{proof}

The graph found by Algorithm~\ref{alg:tree_decomp} may contain edges that are not present in the hidden graph. Next, we show how to prune a supergraph to find the hidden graph.

We take advantage of the following fact about existence of small balanced separators in bounded treewidth graphs (see, e.g., Lemma 7.19 of the textbook of Cygan et al.\ \cite{Cygan15}):
\begin{fact}[folklore]
    Let $G$ be a graph with treewidth $k$. Then there exists a $1/2$-balanced separator of order at most $k+1$.
\end{fact}

A {\em coloring} of a graph $G$ is a mapping from the vertex set to some number of colors such that neighboring vertices have always different color. Graphs with small treewidth can be colored with few colors (see, e.g., Sopena \cite{Sopena97}):
\begin{fact}[folklore]
    Every graph with treewidth at most $k$ can be colored with $k+1$ colors.
\end{fact}

The idea of the algorithm is to find small separators that split the graph into smaller connected components and then recursively to split the components into even smaller ones. At each stage, we find all the neighbors of the vertices belonging to the separator. Then, we continue to split the graph into smaller connected components. At each stage, we keep track of the vertices whose neighbors have been found and exclude them from the queries which allows us to conduct tests for several connected components in parallel.

\begin{algorithm}[ht]
    \caption{An algorithm for \gr given a supergraph of the hidden graph.}\label{alg:superstructure}
    \KwData{A vertex set $V$, a CC oracle $\CC(\cdot)$, a supergraph of the hidden graph $\supgraph$}
    \KwResult{The hidden graph $G$}
    $G^* \gets \supgraph$\;
    Find a tree decomposition $(X, T)$ of $\supgraph$\;
    Color vertices in $\supgraph$ arbitrarily using $k+1$ colors\;
    $R \gets V$ // Vertices that have not been processed yet\;
    $\mathcal{Q} \gets \emptyset$ // Queries that will be conducted\;
    \While{$R \neq \emptyset$}{
        $\mathcal{S} = \{\}$ // Collection of sets\;
        \For{C \textup{that is a connected component of $\supgraph [R]$}}{
            Find a balanced separator $X'$ of $C$\;
            Append $X'$ to $\mathcal{S}$\;
            }
        $h = \max_{X'\in \mathcal{S}} |X'|$\;
        Create a collection $K_1$ of $h(h-1)/2$ sets such that every set in $K_1$ contains at most two elements from $X'$ for all $X'\in \mathcal{S}$ and every pair $u, v\in X'$, $u\neq v$ appears together in at least one set in $K_1$\;
        $S' \gets \bigcup_{U\in \mathcal{S}} U$\;
        Create a collection $K_2$ of $h(k+1)$ sets such that each set in $K_2$ contains one element from each $X'\in \mathcal{S}$ and all vertices of one color in $R\setminus S'$ and every vertex in $S'$ appears at least once in a set with each color\;
        Add all sets in $K_1$ and $K_2$ to $\mathcal{Q}$\;
        $R \gets R \setminus S'$\;
    }
    \For{$Q\in \mathcal{Q}$}{
        $\mathcal{C} \gets \CC(Q)$\;
        \For{$u, v\in Q$}{
            \If{$u$ \textup{and} $v$ \textup{in different connected components in} $\mathcal{C}$}{Remove $uv$ from $G^*$}
        }
    }
    \Return $G^*$
\end{algorithm}

\begin{lemma} \label{lemma:superstructure}
    Suppose that we are given a graph $\supgraph$ whose treewidth is $k$ and which is a supergraph of the hidden graph $G$. Then Algorithm~\ref{alg:superstructure} is a deterministic, non-adaptive algorithm that finds the hidden graph $G$ with $\OO(k^2 \log n)$ CC queries.
\end{lemma}
\begin{proof}
    We show that, by conducting all queries in $\mathcal{Q}$ after an iteration of the \verb|while| loop, we find all the neighbors of the vertices in $V\setminus R$. The algorithm terminates when $R=\emptyset$. Thus, the algorithm constructs the hidden graph $G$.

    At each iteration of the \verb|while| loop, we find all the neighbors to vertices in $S'$. To find out whether $u$ is a neighbor of $v$, it suffices to conduct a query $\CC (Q)$ such that $u, v\in Q$ and none of the other neighbors of $v$ are in $Q$.

    Let us consider an arbitrary vertex $v\in S'$ which belongs to a balanced separator $X'$ in a connected component $C$ in $\supgraph [R]$. First, suppose $u\in X'$. We want to conduct a query with set $Q$ such that $u,v \in Q$ and none of the other neighbors of $v$ are in $Q$. The set $K_1$ contains such a set $Q$ as the collection has at least one set with that contains both $u$ and $v$ and all the other vertices in that set are in different connected components in $\supgraph [R]$.

    Next, suppose that $u\not\in X'$. Now the collection $K_2$ contains a desired set $Q$. The vertices in outside $S'$ included in the set $Q$ are of the same color. Thus, they do not have neighbors that is not in $S'$ is included in the set $Q$. Thus, if there is no edge between $u$ and $v$ in $G$, then $u$ and $v$ have to be in different connected components in $\CC(Q)$.
    
    At each iteration of the \verb|while| loop, each of the connected components have size at most half of the size of the component before removing the balanced separator. Therefore, the procedure terminates after $\OO(\log n)$ iterations. We conduct at most $\OO(k^2)$ queries on each round. Thus, in total, we conduct at most $\OO(k^2 \log n)$ queries. 
\end{proof}

To solve \gr for graphs with bounded treewidth, we simply run Algorithms~\ref{alg:tree_decomp} and \ref{alg:superstructure} resulting in Algorithm~\ref{alg:gr_tw}.
\begin{algorithm}
    \caption{An algorithm for \gr for graphs with bounded treewidth}\label{alg:gr_tw}
    \KwData{A vertex set $V$, a CC oracle $\CC(\cdot)$, treewidth $k$}
    \KwResult{The hidden graph $G$ with high probability}

    $\supgraph \gets \text{ Algorithm~\ref{alg:tree_decomp}} (V, \CC (\cdot ), k)$\;
    $G \gets  \text{ Algorithm~\ref{alg:superstructure}} (V, \CC (\cdot ), \supgraph)$\;
    \Return $G$
\end{algorithm}

\begin{theorem} \label{thm:tw_main}
    Algorithm~\ref{alg:gr_tw} is an adaptive randomized algorithm that solves \gr with $\OO (k^2 \log n)$ CC queries when $k$ is known. The output of the algorithm is correct with high probability.
\end{theorem}
\begin{proof}
    By Lemma~\ref{lemma:tree_decomp}, Algorithm~\ref{alg:tree_decomp} returns with high probability a supergraph of $G$ with treewidth at most $k$ using $\OO (k^2\log n)$ queries. By Lemma~\ref{lemma:superstructure}, Algorithm~\ref{alg:superstructure} always returns the hidden graph $G$ using $\OO (k^2 \log n)$ queries, assuming that $G^*$ has treewidth at most $k$. Combining these two results, we get that Algorithm~\ref{alg:gr_tw} solves \gr with high probability and uses $\OO (k^2 \log n)$ queries.
\end{proof}

\subsubsection{Derandomization}

We can derandomize Algorithm~\ref{alg:gr_tw} similarly to how Algorithm~\ref{alg:max_deg_randomized} is derandomized. That is, we replace the randomized subsets of Algorithm~\ref{alg:tree_decomp} by a $k$-Query-Scheme,
resulting in a deterministic construction of a supergraph with treewidth at most $k$.

\begin{theorem} \label{thm:tw_derandom}
    There exists an adaptive deterministic algorithm that solves \gr with $\OO(k^3 \log n)$ CC queries when $k$ is known.
\end{theorem}
\begin{proof}
    By Lemma~\ref{lemma:k-query}, the $k$-Query-Scheme consists of $\OO (k^3 \log n/k)$ sets that are used as CC queries. Furthermore, by Lemma~\ref{lemma:superstructure}, Algorithm~\ref{alg:superstructure} conducts $\OO (k^2 \log n)$ CC queries. By combining these results, we conclude that the algorithm conducts $\OO (k^3 \log n)$ CC queries.
\end{proof}

\subsubsection{Unknown $k$}

Algorithm~\ref{alg:gr_tw} assumes that the treewidth $k$ of the hidden graph is known. Next, we show how to adapt it to work when $k$ is not known. The proof follows closely the proof of  Theorem~\ref{thm:max_deg_unknown}.
\begin{theorem} \label{thm:tw_unknown}
    There exists an adaptive randomized algorithm and an adaptive deterministic algorithm that solve \gr with $\OO (k^2 \log n)$ and $\OO (k^3\log n)$ CC queries, respectively. The algorithms use $\OO (\log k)$ rounds of adaptivity. The output of the randomized algorithm is correct with high probability. 
\end{theorem}
\begin{proof}
    We start by guessing that $D=1$ us an upper bound for $k$, and then run Algorithm~\ref{alg:tree_decomp} or the algorithm from Theorem~\ref{thm:tw_derandom} as if the known value of $k$ was $D$. If $D$ is smaller than $k$, then the graph outputted has treewidth larger than $D$. In this case, we double $D$ and run the algorithm again as $k$ was $D$, repeating the process $\OO (\log k)$ times.

    Each round of adaptivity performs $\OO (D^2 \log n)$ or $\OO (D^3 \log n)$ queries. Since $D$ is doubled each round, the total number of queries is dominated by the last round, resulting in the total number of queries being $\OO (k^2 \log n)$ or $\OO (k^3 \log n)$. 

    For the randomized algorithm, it remains to show that the output is correct with high probability. We notice that Algorithm~\ref{alg:tree_decomp} returns a supergraph of the hidden graph. Thus, if $D$ is smaller than $k$, then we double $D$. It follows that at the final round $D$ is larger than or equal to $k$. Now, by Lemma~\ref{lemma:tree_decomp}, Algorithm~\ref{alg:tree_decomp} returns, with high probability, a graph whose treewidth is at most $2k$. 
\end{proof}

Combining the adaptive algorithms for $m$, $\Delta$, and $k$ yields the following result:
\begin{theorem}\label{thm:combination}
    There exists an adaptive randomized algorithm that solves \gr with $\OO(\min\{m/\log m, \Delta^2, k^2\} \cdot \log n)$ CC queries. The output of the algorithm is correct with high probability. 
\end{theorem}
\begin{proof}
    We use the adaptive randomized algorithms from Corollary~\ref{thm:edge}, Theorem~\ref{thm:max_deg_unknown}, and Theorem~\ref{thm:tw_unknown}. If ran separately, they would stop after performing some number of queries---say, $t_1$, $t_2$, and $t_3$, respectively. Consider instead running the algorithms in parallel, such that in the $i$th iteration we perform the $i$th query of each of the three algorithms. If any of them would stop after performing that query, we output the graph computed by that algorithm. In total, we thus perform $3 \cdot \min\{t_1, t_2, t_3\} = \OO(\min\{m / \log m, \Delta^2, k^2\} \cdot \log n)$ CC queries. Further, since the output of each algorithm is correct with high probability, so is the output of this algorithm.
\end{proof}

\subsection{Bounded Degeneracy}\label{sec:degeneracy}

Next, we prove a polylogarithmic upper bound for the query complexity for graphs of bounded degeneracy. 
The bound holds also for graphs with bounded number of edges, maximum degree, and treewidth, since such graphs also have bounded degeneracy. For planar graphs, the degeneracy~$d$ is at most five \cite{Lick70} but the number of edges, the maximum degree and the treewidth can be unbounded in terms of $d$, so the polylogarithmic upper bound holds for a strictly larger family of graphs than with the previous algorithms. However, this comes at the expense of an additional $\OO(\log n)$ factor in the query complexity.

Note that if the degeneracy is $d$, then the average degree of every induced subgraph is at most $2d$. Consequently, if the average degree of a graph is at most $2d$, then at least half of its vertices have degree at most $4d$. Our approach is then to repeatedly remove all vertices of degree at most $4d$ from $G$, at least halving the number of vertices in each iteration.

\begin{theorem}\label{thm:degeneracy}
    Algorithm~\ref{alg:degeneracy} is an adaptive randomized algorithm that solves \gr with $\OO(d^2 \log^2 n)$ CC queries when $d$ is known. The output of the algorithm is correct with high probability.
\end{theorem}
\begin{proof}
    Analogously to the proof of Theorem~\ref{thm:max_deg_randomized}, one can show that after running Algorithm~\ref{alg:max_deg_randomized} with a possibly incorrect value $4d$ for the maximum degree, the outputted graph $G'$ has an edge $uv$ if and only if the hidden graph has that edge with high probability for all vertices $v$ with degree at most $4d$ in $G$. In other words, with high probability, we correctly identify all neighbors of vertices with degree at most $4d$ in $G$. Denote this set of vertices by $L$. It then remains to identify the structure of the subgraph of the hidden graph induced by $V \setminus L$. This is precisely what Algorithm~\ref{alg:degeneracy} does, repeating the process until we have identified the neighbors for all vertices. Since at least half of the vertices of each induced subgraph have degree at most $4d$, this process gets repeated at most $1 + \log_2 n$ times. Taking into account the query complexity of Algorithm~\ref{alg:max_deg_randomized}, we perform $\OO(d^2 \log^2 n)$ CC queries in total. Finally, recall from the proof of Theorem~\ref{thm:max_deg_randomized} that Algorithm~\ref{alg:max_deg_randomized} works correctly with probability at least $1 - 1/n$. Since $(1 - 1/n)^{1 + \log_2 n}$ tends to $1$ as $n$ increases, Algorithm~\ref{alg:degeneracy} works correctly with high probability.
\end{proof}

\begin{algorithm}[ht]
    \caption{An algorithm for \gr for graphs with bounded degeneracy}\label{alg:degeneracy}
    \KwData{A vertex set $V$, a CC oracle $\CC(\cdot)$, degeneracy $d$}
    \KwResult{The hidden graph $G$ with high probability}
    Initialize $\supgraph$ to a complete graph on $V$\;
    $V' \gets V$\;
    \While{$V' \neq \emptyset$}{
        Let $G'$ be the output of Algorithm~\ref{alg:max_deg_randomized} on the vertex set $V'$ with (incorrect) maximum degree $4d$\;
        Remove edges $uv$ from $\supgraph$ for which edge $uv$ is not in $G'$ and $u, v \in V'$\;
        Let $L$ be the set of vertices in $G'$ with degree at most $4d$\;
        $V' \gets V' \setminus L$\;
    }
    \Return $\supgraph$\;
\end{algorithm}

\subsubsection{Derandomization}

If we instead apply the non-adaptive deterministic algorithm from Theorem~\ref{thm:max_deg_deterministic} instead of the randomized algorithm for bounded degree graphs, we get an adaptive deterministic algorithm for graphs with bounded degeneracy.
\begin{theorem}\label{thm:degeneracy_deterministic}
    There exists an adaptive deterministic algorithm that solves \gr with $\OO\big(d^3 \log(n) \log(n / d)\big)$ CC queries when $d$ is known.
\end{theorem}
\begin{proof}
    The proof is analogous to the proof of Theorem~\ref{thm:degeneracy} except that we use a $d$-Query-Scheme similarly to Theorem~\ref{thm:max_deg_deterministic} to eliminate randomness.
\end{proof}

\subsubsection{Unknown $d$}

The idea for adjusting the algorithm for graphs of unknown parameter value works similarly to before, but some care is needed to identify that our guess for the degeneracy was too low.

\begin{theorem}\label{thm:degeneracy_unknown}
    There exist an adaptive randomized algorithm and an adaptive deterministic algorithm that solve \gr with $\OO\big(d^2 \log^2 n\big)$ and $\OO\big(d^3 \log(n) \log(n / d)\big)$ CC queries, respectively. The output of the randomized algorithm is correct with high probability.
\end{theorem}
\begin{proof}
    We start by guessing that $D = 1$ is an upper bound for $d$, and then run Algorithm~\ref{alg:degeneracy} or the algorithm from Theorem~\ref{thm:degeneracy_deterministic} as if the known value of $d$ were $D$. If we use the former algorithm and it runs correctly or we use the latter algorithm, then we identify that at least half of the vertices have degree at most $4(D+1)$ in each iteration if $D \ge d$. This may also occur if $D < d$, but if does not occur, then we know that $D < d$ and we double the value of $D$ and restart the algorithm. The query complexity is again dominated by the last round, for which $D < 2d$.

    It remains to show the high probability of the correctness of the randomized algorithm. Recall that the probability of Algorithm~\ref{alg:degeneracy} succeeds is at least $(1 - 1/n)^{1 + \log_2 n}$, and note that if the algorithm succeeds on all $\OO(\log d)$ rounds of adaptivity, then the output is correct. The probability of this is at least $(1 - 1/n)^{\OO(\log^2 n)} \approx (1/\mathrm{e})^{\OO(\log^2(n) / n)}$, which tends to $1$ as $n$ increases.
\end{proof}

\section{Lower Bounds}\label{sec:lower_bounds}

In this section, we show that the query complexity of our randomized algorithms cannot be improved by a factor larger than $\OO(\log n)$ with a proof similar to the lower bound proof of Konrad et al.~\cite{Konrad24}.
Whilst Black et al.~\cite{Black25} have showed a tighter upper bound with respect to $m$ for the $\#$CC oracle, their information-theoretical argument does not work here because our oracle gives more bits of information than theirs. 
However, we confirm that their non-adaptive lower bound is applicable for the CC oracle, showing that non-adaptive algorithms require $\Omega(n^2)$ CC queries if they are not given additional information about the structure.

\begin{theorem}\label{thm:lower_bound_mkd}
    No algorithm can solve \gr with $o(m)$, $o(k^2)$, $o(\Delta^2)$, or $o(d^2)$ CC queries with probability greater than $1/2$.
\end{theorem}
\begin{proof}
    Fix two disjoint subsets of vertices $C_1, C_2 \subseteq V$ of some size $\eta$ and consider the family of graphs $\mathcal{G}_{\eta}$ where both $C_1$ and $C_2$ induce a clique and there is an arbitrary subset of edges between $C_1$ and $C_2$. Clearly, $|\mathcal{G}_{\eta}| = 2^{\eta^2}$. 

    Suppose now that there is a randomized algorithm that outputs the correct graph for each hidden graph $G \in \mathcal{G}_{\eta}$ with probability greater than $1/2$ while performing at most $\eta^2 - 1$ queries. By Yao's lemma, there is then a deterministic algorithm that performs $\eta^2 - 1$ queries and outputs the correct graph for more than half of the hidden graphs $G \in \mathcal{G}_{\eta}$. 
    However, observe that for any CC query on a set $S \subseteq V$, there are exactly two possible outputs over the graphs $G \in \mathcal{G}_{\eta}$: the vertices in $S \setminus (C_1 \cup C_2)$ are always in connected components of size $1$ and the remaining connected components induced by $S$ are either $S \cap C_1$ and $S \cap C_2$ or just $S$. Consequently, the deterministic algorithm can identify in $\eta^2 - 1$ queries only $2^{\eta^2 - 1} \le |\mathcal{G}_{\eta}|/2$ hidden graphs, resulting in a contradiction. Thus, $\Omega(\eta^2)$ queries are needed. We have $k, \Delta, d \le 2\eta$ and $m \le 2\eta^2$ for this family of hidden graphs, completing the proof.  
\end{proof}

\begin{corollary}[\cite{Black25}]\label{cor:lower_bound_n}
    No non-adaptive algorithm can solve \gr with $o(n^2)$ CC queries with probability greater than $1/2$ even if $k=2$ and $m = \OO(n)$.
\end{corollary}
\begin{proof}
    Let $G_{uv}$ denote a graph on $n$ vertices such that we start by having the vertices $u$ and $v$, and then add $n - 2$ vertices that are connected to $u$ and $v$. Denote by $G_{uv}^+$ the graph obtained from $G_{uv}$ by adding the edge $uv$. Clearly, $G_{uv}$ and $G_{uv}^+$ have treewidth $2$ and $m = \OO(n)$.

    Suppose that the hidden graph is either $G_{uv}$ or $G_{uv}^+$ but the non-adaptive algorithm does not perform the query $\{u, v\}$. Then, we cannot distinguish between them since $u$ and $v$ would not be in distinct connected components for any of the performed CC queries. After performing all the queries, the best one can do is to guess between $G_{uv}$ and $G_{uv}^+$, resulting in a success probability $1/2$. For the algorithm to work correctly with probability greater than $1/2$ for any graph $G_{uv}$ or $G_{uv}^+$, the non-adaptively constructed set of queries $\mathcal{Q}$ has to include the query $\{u, v\}$ for all distinct vertices $u, v \in V$ with probability greater than $1/2$. By the linearity of expectation, the expected size of $\mathcal{Q}$ would then be
    \[ \sum_{\substack{u, v \in V\\u \neq v}} \Pr(\{u, v\} \in \mathcal{Q}) > \sum_{\substack{u, v \in V\\u \neq v}} \frac{1}{2} = \binom{n}{2} / 2. \]
    In other words, for any algorithm that always performs $o(n^2)$ queries, there is an instance on which it fails with probability at least $1/2$.
\end{proof}

\section{Relationship with Other Oracles}

In this section, we compare our CC oracle to the connected component counting oracle, maximal independent set oracle, and the separation oracle.

\subsection{Relationship with the Connected Component Counting Oracle}

We start by proving that it is not sufficient to use the {connected component counting}~($\#$CC) \emph{oracle} of Black et al.~\cite{Black25} for obtaining our results, since $\Omega(n / \log n)$ $\#$CC queries are needed for already very constrained graph families. The main observation for this result is that the $\#$CC oracle has $n + 1$ possible distinct values it can output, and so in $t$ queries we can correctly identify only $(n + 1)^t$ graphs. By constructing a family of graphs larger than that, we can show that more than $t$ queries are needed. Similar combinatorial arguments were used by Michel and Scott~\cite{Michel25} and Black et al.~\cite{Black25} to prove lower bounds for the maximal independent set oracle in bounded degree graphs and for the $\#$CC oracle in graphs with $m$ edges.

\begin{theorem}\label{thm:counting-is-weak}
    No algorithm can solve all instances of \gr with maximum degree, treewidth, or degeneracy at most $p$ with $\Omega(np / \log n)$ $\#$CC queries with probability greater than $1/2$.
\end{theorem}
\begin{proof}
    Without loss of generality, assume that $p$ divides $n$. Let $G$ be a graph on $n$ vertices $v_1, v_2, \dots, v_n$ where there is an edge $v_iv_j$ if and only if $\lfloor(i - 1) / p\rfloor = \lfloor(j - 1) / p\rfloor$ or $i = j - 1$. The maximum degree, treewidth, and degeneracy of $G$ are all at most $p$.
    
    We can now construct a family of graphs $\mathcal{G}$ by considering all subgraphs of $G$ that retain all edges of the form $v_iv_{i+1}$. The size of the family is then
    \[ |\mathcal{G}| = 2^{\binom{p - 1}{2} \cdot \frac{n}{p}} = (n + 1)^{\binom{p - 1}{2} \cdot \frac{n}{p} \cdot \log_{n + 1}(2)}. \]

    Suppose now that there is a randomized algorithm that outputs the correct graph for each hidden graph $G' \in \mathcal{G}$ with probability greater than $1/2$ while performing at most $t \coloneqq \binom{p - 1}{2} \cdot \frac{n}{p} \cdot \log_{n + 1}(2) - 1$ $\#$CC queries. Then, by Yao's lemma, there is a deterministic algorithm that performs $t$ and outputs the correct graph for more than half of the hidden graphs $G' \in \mathcal{G}$. However, a deterministic algorithm can only have $(n + 1)^t \le |\mathcal{G}| / (n+1)$ distinct outputs if it performs $t$ queries, and so the output is incorrect for more than half of the hidden graphs.
\end{proof}

\subsection{Relationship with the Maximal Independent Set Oracle}

We next compare the CC oracle against the {\em maximal independent set} (MIS) \emph{oracle} of Konrad et al.~\cite{Konrad24}, which returns a maximal independent set of $G[S]$. Note that the MIS oracle can be adversarial, that is, it may return the ``least useful'' maximal independent set.
We start this section by demonstrating that CC queries are more powerful than MIS queries on some instances.

\begin{theorem} \label{thm:cc_powerful}
    There are instances of \gr that take $\Omega(n)$ MIS queries to solve but $\OO(\log n)$ CC queries.
\end{theorem}
\begin{proof}
    Take a star graph on $n$ vertices centered around the vertex $v$ and consider query sets $Q_1, Q_2, \dots, Q_\ell$. Suppose our adversarial MIS oracle returns the set $Q_i \setminus \{v\}$ for the query set $Q_i$ for all $i \in \{1, 2, \dots, \ell\}$. Then, the results are also consistent with a graph with edges only between $v$ and vertices $\max_{u \in Q_i \setminus \{v\}} u$ for all $i$, under some arbitrary ordering of the vertices. Thus, we need $\ell$ to be at least $n-1$ to recover the original graph. Meanwhile, the graph is a tree, and thus the instance is solvable by $\OO(\log n)$ CC queries by Theorem~\ref{thm:tw_derandom}.
\end{proof}

The following theorems show that a MIS oracle cannot be simulated with a CC oracle without calling it polynomially many times in the size of the queried set, and vice versa.

\begin{theorem}\label{thm:sim_cc}
    We can compute an MIS on $S$ with $\OO(|S|)$ CC queries, and there are instances where $\Omega(|S|)$ queries are needed.
\end{theorem}
\begin{proof}
    For the upper bound, we maintain an independent set $T$. For each node $v \in S$, try adding it to $T$ and test if it is still an independent set, that is, are all vertices of $T$ in distinct connected components.

    For the lower bound, assume that the subgraph induced by $S$ is a clique. All independent sets are thus of size $1$. To simulate an MIS oracle by outputting some set $\{v\} \subseteq S$, our algorithm would need to ensure all the other vertices in $S$ are neighbors of $v$. However, if there is any pair of vertices $\{u, v\} \subseteq S$ that is not queried, then the outputs for the CC queries would also be consistent with a graph obtained by taking a clique on $S$ and removing the edge $uv$. Hence, we cannot ensure that $\{v\}$ is an MIS without performing $|S| - 1$ CC queries.
\end{proof}

\begin{theorem}\label{thm:sim_mis}
    We can compute the connected components on $S$ with $\OO(|S|^2)$ MIS queries, and there are instances where $\Omega(|S|^2)$ queries are needed.
\end{theorem}
\begin{proof}
    For the upper bound, query all pairs. For the lower bound, assume that the subgraph induced by $S$ consists of two disjoint cliques $S_1$ and $S_2$ of size $|S|/2$. Suppose that we have not performed a MIS query on some pair $\{s_1, s_2\}$ with $s_1 \in S_1$ and $s_2 \in S_2$. Then, an adversarial oracle can always return a MIS that does not contain both $s_1$ and $s_2$, meaning that we cannot be certain that there is no edge between the two nodes. Since there are $\Theta(|S|^2)$ such pairs, we need that many MIS queries to be able to ensure that $S_1$ and $S_2$ are not connected.
\end{proof}

\subsection{Relationship with the Separation Oracle}\label{sec:separation}

Finally, we compare the CC oracle against the separation oracle used in machine learning. We also discuss implications of known results for the separation oracle.

The {\em separation} (Sep) {\em oracle} has three inputs: A pair of vertices $v$ and $w$ and a set of vertices $U \subseteq V \setminus\{v, w\}$. The oracle returns \verb|Yes| if $v$ and $w$ are in different connected components in the induced subgraph $G[V\setminus U]$. Otherwise, the oracle returns \verb|No|.

It is clear that the CC oracle is more powerful than the Sep oracle. Theorem~\ref{thm:sep_lb} shows that each Sep query can be answered by calling one CC query. On the other hand, Theorem~\ref{thm:sep_ub} that sometimes one needs $\Omega (|S|^2)$ Sep queries to answer one CC query.

\begin{theorem} \label{thm:sep_lb}
    We can compute a separation query $\Sep (v, w, U)$ with one CC query.
\end{theorem}
\begin{proof}
   We can use the following procedure: Call $\CC (V\setminus U)$. If $v$ and $w$ are in different connected components, return \verb|Yes|, otherwise return \verb|No|.
\end{proof}

\begin{theorem} \label{thm:sep_ub}
    We can compute connected components on $S$ with $\OO (|S|^2)$ Sep queries and there are instances where $\Omega (|S|^2)$ queries are needed.
\end{theorem}
\begin{proof}
    For the upper bound, use the following procedure: For all pairs $v, w \in S$, query $\Sep (v, w, V\setminus S)$. If the answer is \verb|Yes|, then $v$ and $w$ are in the same connected component. Otherwise, they are in different connected components.

    For the lower bound, assume that the induced subgraph $G[V \setminus S]$ is an empty graph. However, if there is any pair of vertices $v$ and $w$ such that $\Sep (v, w, V\setminus S)$ is not queried, then the outputs from the separation queries would also be consistent with a graph with the edge $vw$. Hence, we cannot ensure that we have found all connected components without performing $\Omega (|S|^2)$ Sep queries.
\end{proof}

It is trivial to solve the \gr problem with $\OO(n^2)$ Sep queries. We just perform queries $\Sep (v, w, V\setminus \{v, w\})$ for each pair of vertices $v, w\in V, v\neq w$.

When separation queries are used in structure learning in Markov networks, one is not solely interested in the number of queries but also their sizes. Here, the size of a query~$\Sep (v, w, U)$ is $|U|$. This is motivated by the fact that in practice the oracle is replaced by a statistical test and the amount of data needed for reliable tests grows exponentially with respect to the size of the test. That is, one often wants to conduct only small tests, even if the number of tests increases.

Korhonen et al. \cite{Korhonen24} have shown that when the hidden graph has maximum pairwise connectivity $\lambda$, then one needs always to conduct at least one Sep query with size at least~$\lambda$. This result translates directly to the CC oracle. If the hidden graph has maximum pairwise connectivity $\lambda$ then the \gr problem cannot be solved by conducting queries for a scheme~$Q = \{Q_1, \ldots, Q_\ell\}$ if $\min_i |Q_i| > \lambda$.

\section{Conclusions}

We proposed a novel connected components oracle and showed that it requires only a (poly)logarithmic number of queries to solve graph reconstruction in graphs of bounded number of edges, maximum degree, treewidth, or degeneracy. These results were complemented by unconditional lower bounds and comparisons against the $\#$CC oracle, the MIS oracle, and the $\Sep$ oracle.

It remains open whether the gap between our lower and upper bounds can be tightened. Michel and Scott~\cite{Michel25} recently improved the lower bound for the MIS oracle from $\Omega(\Delta^2)$ to $\omega(\Delta^2)$, but their techniques do not appear to directly generalize for the CC oracle.

Another direction is to study how much weaker the CC oracle could be made such that we would still obtain similar upper bounds. For example, an oracle that might be weaker than the CC oracle but stronger than the $\#$CC oracle would be one that could answer the following query: Given a set $Q$, output one vertex from each connected component of $G[Q]$.

In the context of Markov networks where the separation oracle is used, the query complexity of graph reconstruction has been recently studied under constraints on the size of the queried sets \cite{Korhonen24}. For example, $n^{\Omega(\lambda)}$ separation queries of size $\lambda$ are needed when the maximum pairwise connectivity is $\lambda$. Here, we were able to reconstruct the graph in $\OO(\lambda^2 \log n)$ queries, but the queried sets were larger. A natural follow-up question would thus be to study how the connected components oracle performs with query sets of bounded size.

\bibliography{arxiv}

\begin{thebibliography}{10}

\bibitem{Abasi19}
Hasan Abasi and Nader~H. Bshouty.
\newblock On learning graphs with edge-detecting queries.
\newblock In {\em Algorithmic Learning Theory, {ALT} 2019}, volume~98 of {\em
  Proceedings of Machine Learning Research}, pages 3--30. {PMLR}, 2019.

\bibitem{Abrahamsen16}
Mikkel Abrahamsen, Greg Bodwin, Eva Rotenberg, and Morten St{\"{o}}ckel.
\newblock Graph reconstruction with a betweenness oracle.
\newblock In {\em Proceedings of the 33rd Symposium on Theoretical Aspects of
  Computer Science, {STACS} 2016}, volume~47 of {\em LIPIcs}, pages 5:1--5:14.
  Schloss Dagstuhl - Leibniz-Zentrum f{\"{u}}r Informatik, 2016.

\bibitem{Alon05}
Noga Alon and Vera Asodi.
\newblock Learning a hidden subgraph.
\newblock {\em {SIAM} J. Discret. Math.}, 18(4):697--712, 2005.

\bibitem{Alon04}
Noga Alon, Richard Beigel, Simon Kasif, Steven Rudich, and Benny Sudakov.
\newblock Learning a hidden matching.
\newblock {\em {SIAM} J. Comput.}, 33(2):487--501, 2004.

\bibitem{Alon95}
Noga Alon, Raphael Yuster, and Uri Zwick.
\newblock Color-coding.
\newblock {\em J. {ACM}}, 42(4):844--856, 1995.

\bibitem{Angluin08}
Dana Angluin and Jiang Chen.
\newblock Learning a hidden graph using ${O}(\log n)$ queries per edge.
\newblock {\em J. Comput. Syst. Sci.}, 74(4):546--556, 2008.

\bibitem{Beerliova06}
Zuzana Beerliova, Felix Eberhard, Thomas Erlebach, Alexander Hall, Michael
  Hoffmann, Mat{\'{u}}s Mihal{\'{a}}k, and L.~Shankar Ram.
\newblock Network discovery and verification.
\newblock {\em {IEEE} J. Sel. Areas Commun.}, 24(12):2168--2181, 2006.

\bibitem{Beigel01}
Richard Beigel, Noga Alon, Simon Kasif, Mehmet~Serkan Apaydin, and Lance
  Fortnow.
\newblock An optimal procedure for gap closing in whole genome shotgun
  sequencing.
\newblock In {\em Proceedings of the Fifth Annual International Conference on
  Computational Biology, {RECOMB} 2001}, pages 22--30. {ACM}, 2001.

\bibitem{Black25}
Hadley Black, Arya Mazumdar, Barna Saha, and Yinzhan Xu.
\newblock Optimal graph reconstruction by counting connected components in
  induced subgraphs.
\newblock In {\em The Thirty Eighth Annual Conference on Learning Theory,
  {COLT} 2025}, volume 291 of {\em Proceedings of Machine Learning Research},
  pages 315--343. {PMLR}, 2025.

\bibitem{Bouvel05}
Mathilde Bouvel, Vladimir Grebinski, and Gregory Kucherov.
\newblock Combinatorial search on graphs motivated by bioinformatics
  applications: {A} brief survey.
\newblock In {\em Proceedings of the 31st International Workshop
  Graph-Theoretic Concepts in Computer Science, {WG} 2005}, volume 3787 of {\em
  Lecture Notes in Computer Science}, pages 16--27. Springer, 2005.

\bibitem{Chow68}
C.~K. Chow and C.~N. Liu.
\newblock Approximating discrete probability distributions with dependence
  trees.
\newblock {\em {IEEE} Trans. Inf. Theory}, 14(3):462--467, 1968.

\bibitem{Cygan15}
Marek Cygan, Fedor~V. Fomin, Lukasz Kowalik, Daniel Lokshtanov, D{\'{a}}niel
  Marx, Marcin Pilipczuk, Michal Pilipczuk, and Saket Saurabh.
\newblock {\em Parameterized Algorithms}.
\newblock Springer, 2015.

\bibitem{Grebinski98}
Vladimir Grebinski and Gregory Kucherov.
\newblock Reconstructing a hamiltonian cycle by querying the graph: Application
  to {DNA} physical mapping.
\newblock {\em Discret. Appl. Math.}, 88(1-3):147--165, 1998.

\bibitem{Grebinski00}
Vladimir Grebinski and Gregory Kucherov.
\newblock Optimal reconstruction of graphs under the additive model.
\newblock {\em Algorithmica}, 28(1):104--124, 2000.

\bibitem{Kannan18}
Sampath Kannan, Claire Mathieu, and Hang Zhou.
\newblock Graph reconstruction and verification.
\newblock {\em {ACM} Trans. Algorithms}, 14(4):40:1--40:30, 2018.

\bibitem{Karger01}
David~R. Karger and Nathan Srebro.
\newblock Learning {M}arkov networks: maximum bounded tree-width graphs.
\newblock In {\em Proceedings of the Twelfth Annual Symposium on Discrete
  Algorithms, SODA 2001}, pages 392--401. {ACM/SIAM}, 2001.

\bibitem{Koller09}
Daphne Koller and Nir Friedman.
\newblock {\em Probabilistic Graphical Models - Principles and Techniques}.
\newblock {MIT} Press, 2009.

\bibitem{Konrad24}
Christian Konrad, Conor O'Sullivan, and Victor Traistaru.
\newblock Graph reconstruction via {MIS} queries.
\newblock In {\em Proceedings of the 16th Conference on Innovations in
  Theoretical Computer Science, {ITCS} 2025}, volume 325 of {\em LIPIcs}, pages
  66:1--66:19. Schloss Dagstuhl - Leibniz-Zentrum f{\"{u}}r Informatik, 2025.

\bibitem{Korhonen24}
Tuukka Korhonen, Fedor~V. Fomin, and Pekka Parviainen.
\newblock Structural perspective on constraint-based learning of {M}arkov
  networks.
\newblock In {\em International Conference on Artificial Intelligence and
  Statistics, {AISTATS} 2024}, volume 238 of {\em Proceedings of Machine
  Learning Research}, pages 1855--1863. {PMLR}, 2024.

\bibitem{Lick70}
Don~R. Lick and Arthur~T. White.
\newblock k-degenerate graphs.
\newblock {\em Canadian Journal of Mathematics}, 22(5):1082–1096, 1970.

\bibitem{Mathieu23}
Claire Mathieu and Hang Zhou.
\newblock A simple algorithm for graph reconstruction.
\newblock {\em Random Struct. Algorithms}, 63(2):512--532, 2023.

\bibitem{Michel25}
Lukas Michel and Alex~D. Scott.
\newblock Lower bounds for graph reconstruction with maximal independent set
  queries.
\newblock {\em Theor. Comput. Sci.}, 1034:115121, 2025.

\bibitem{Porat11}
Ely Porat and Amir Rothschild.
\newblock Explicit nonadaptive combinatorial group testing schemes.
\newblock {\em IEEE Transactions on Information Theory}, 57(12):7982--7989,
  2011.

\bibitem{Reyzin07}
Lev Reyzin and Nikhil Srivastava.
\newblock Learning and verifying graphs using queries with a focus on edge
  counting.
\newblock In {\em Proceedings of the Eighteenth International Conference on
  Algorithmic Learning Theory, {ALT} 2007}, volume 4754 of {\em Lecture Notes
  in Computer Science}, pages 285--297. Springer, 2007.

\bibitem{Rong21}
Guozhen Rong, Wenjun Li, Yongjie Yang, and Jianxin Wang.
\newblock Reconstruction and verification of chordal graphs with a distance
  oracle.
\newblock {\em Theor. Comput. Sci.}, 859:48--56, 2021.

\bibitem{Rong22}
Guozhen Rong, Yongjie Yang, Wenjun Li, and Jianxin Wang.
\newblock A divide-and-conquer approach for reconstruction of $\{C_{\ge
  5}\}$-free graphs via betweenness queries.
\newblock {\em Theor. Comput. Sci.}, 917:1--11, 2022.

\bibitem{Rose74}
Donald~J. Rose.
\newblock On simple characterizations of k-trees.
\newblock {\em Discret. Math.}, 7(3-4):317--322, 1974.

\bibitem{Sen10}
Sandeep Sen and V.~N. Muralidhara.
\newblock The covert set-cover problem with application to network discovery.
\newblock In {\em Proceedings of the 4th International Workshop on Algorithms
  and Computation, {WALCOM} 2010}, volume 5942 of {\em Lecture Notes in
  Computer Science}, pages 228--239. Springer, 2010.

\bibitem{Sopena97}
{\'{E}}ric Sopena.
\newblock The chromatic number of oriented graphs.
\newblock {\em J. Graph Theory}, 25(3):191--205, 1997.

\end{thebibliography}

\end{document}